\documentclass[hidelinks,11pt,english]{article} 
\usepackage{geometry}                
\usepackage{graphicx}
\usepackage{amssymb}
\usepackage{amsmath}
\usepackage{epstopdf}
\usepackage[dvipsnames]{xcolor}
\usepackage{natbib}
\usepackage{algorithm}
\usepackage{hyperref}
\newcommand{\bean}{\begin{eqnarray*}}
\newcommand{\eean}{\end{eqnarray*}}
\newcommand{\bea}{\begin{eqnarray}}
\newcommand{\eea}{\end{eqnarray}}
\newcommand{\beaa}{\begin{array}}
\newcommand{\eeaa}{\end{array}}
\usepackage{apxproof}
\usepackage{xargs}
\usepackage[colorinlistoftodos,prependcaption,textsize=tiny]{todonotes}
\newcommandx{\unsure}[2][1=]{\todo[fancyline,backgroundcolor=red!25,bordercolor=red,#1]{#2}}
\newcommandx{\change}[2][1=]{\todo[linecolor=blue,backgroundcolor=blue!25,bordercolor=blue,#1]{#2}}
\newcommandx{\info}[2][1=]{\todo[linecolor=OliveGreen,backgroundcolor=OliveGreen!25,bordercolor=OliveGreen,#1]{#2}}
\newcommandx{\improvement}[2][1=]{\todo[linecolor=Plum,backgroundcolor=Plum!25,bordercolor=Plum,#1]{#2}}
\newcommandx{\thiswillnotshow}[2][1=]{\todo[disable,#1]{#2}}

\newcommand{\cR}{\mathcal{R}}
\newcommand{\cQ}{\mathcal{Q}}

\newtheoremrep{theorem}{Theorem}  
\newtheoremrep{lemma}{Lemma}  
\newtheoremrep{conjecture}{Conjecture}
\newtheoremrep{proposition}{Proposition}
\newtheoremrep{corrolary}{Corollary}
\newtheoremrep{example}{Example}

\begin{document}
\title{The Refined Assortment Optimization Problem}

\author{
	Gerardo Berbeglia\footnote{Melbourne Business School, The University of Melbourne, Australia, g.berbeglia@mbs.edu.}
	\and
	Alvaro Flores\footnote{Department of Innovation, HIVERY, Sydney, Australia, Australia, alvaro@hivery.com.}
	\and
	Guillermo Gallego\footnote{Industrial Engineering and Decision Analytics, Hong Kong University of Science and Technology, Kowloon, Hong Kong, ggallego@ust.hks. Supported by RGC project 16211619. }
}
\date{December 23, 2020}
\maketitle

\begin{abstract}

We introduce the refined assortment optimization problem where a firm may decide to make some of its products harder to get instead of making them unavailable as in the traditional assortment optimization problem. Airlines, for example, offer fares with severe restrictions rather than making them unavailable. This is a more subtle way of handling the trade-off between demand induction and demand cannibalization. For the latent class MNL model, a firm that engages in refined assortment optimization can make up to $\min(n,m)$ times more than one that insists on traditional assortment optimization, where $n$ is the number of products and $m$ the number of customer types. Surprisingly, the revenue-ordered assortment heuristic has the same performance guarantees relative to {\em personalized} refined assortment optimization as it does to traditional assortment optimization. Based on this finding, we construct refinements of the revenue-order heuristic and measure their improved performance relative to the revenue-ordered assortment and the optimal traditional assortment optimization problem. We also provide tight bounds on the ratio of the expected revenues for the refined versus the traditional assortment optimization for some well known discrete choice models.

\end{abstract}
\newpage

\section{Introduction}\label{sec:intro}

Discrete choice models are of interest to both industry and academia as they can accurately capture demand substitution patterns that enable firms to offer a better mix of products to consumers, see \cite{thurstone1927law,anderson1992nourl, luce1959, Plackett1975, mcfadden1978modeling, Guadagni1983nourl, mcfadden200mixednourl, ben1985nourl}.  There is a growing body of literature supporting the idea that using discrete choice models  leads to better sales outcomes \citep{talluri2004revenue,Vulcano2010, farias2013nonparametric}.  Parametric discrete choice models yield more accurate estimates than machine learning techniques \cite{feldman2018customer} unless data is abundant and the ground truth model cannot be easily captured by a parametric model  \citep{chen2019use}. Parametric models have the additional advantage that they are amenable to optimization of the firm's objective such as expected sales or expected revenues.

A key problem at the center of e-commerce and revenue management is the traditional assortment optimization problem (TAOP),  which requires finding a subset of products that the firm should offer to maximize expected revenues (or profits), see \cite{talluri2004revenue, KokAssortment,mendez2014branch,Vulcano2010}. The tradeoff is between marginal revenues and demand cannibalization. Indeed, it is optimal to exclude products when their revenue contribution is lower than the revenue losses due to demand cannibalization so the firm can improve profits by redirecting  part of the demand of the excluded products to other more profitable products in the assortment.

In practice, firms may prefer a more subtle approach that make some products harder to get rather than unavailable. This avoids some of the feared cannibalization without completely losing their revenue potential. As an example, a human resource firm in Hong Kong that specializes in placing domestic helpers does not disclose its entire list to customers. Instead, they first offer a list of helpers that have been in the system for a while. Customers that reject the first batch typically get a second list of better qualified helpers. Restricted fares in revenue management are a compromise between offering unrestricted low-fares and not offering them at all. By imposing time-of-purchase and travel restrictions they make these fares less desirable for people who travel for business allowing the airline to obtain revenue from seats that otherwise would fly empty without severely cannibalizing demand for higher-fares. Rolex, makes its steel sport watches extremely difficult to find, steering impatient customers to buy gold watches instead.

The examples above are instances of \emph{refined assortment optimization}, which we will model and study in detail in  this paper.  While our stylized model is general in nature, the exact product modification is context dependent and it can range from temporal delays on delivery or handover of the product, to obscuring or removing some attributes to make some products less desirable. Firms are sometimes forced to display products in a way that makes them less desirable due to a limit on the number of prime display locations. This is known as the product framing problem, see \cite{tversky1989rational}, \cite{feng2007implementing, Craswell_2008nourl,Kempe_2008nourl, aggarwal2008sponsored} and \cite{gallego2020approximation}. Although similar to the refined assortment optimization problem, the crucial difference is that in the product framing problem the {\em refinements} are a consequence of exogenous constraints, whereas these are endogenous decisions in the refined assortment optimization.



\subsection{Contributions}\label{sub:contrib}

This paper is, to our knowledge, the first to study refined assortment optimization.  Our contributions are as follows.


\begin{itemize}
	\item We introduce the \emph{Refined Assortment Optimization Problem} (RAOP) and the \emph{Personalized Assortment Optimization Problem}(p-RAOP) as two continuous optimization problems (Section~\ref{sec:model}).
	\item We show that when consumers follow the latent class multinomial logit (LC-MNL) with $m$ segments and $n$ products, the seller can make up to a factor $\min\{m,n\}$ more with the RAOP than with the TAOP. As a consequence, when consumers follow any random utility model, the seller could make up to $n$ times more by using RAOP instead of TAOP. For the \emph{Random Consideration Set} (RCS) model, the firm can make up to a factor of 2 more with the RAOP relative to the TAOP.  We provide examples to show that these bounds are tight (Section~\ref{sec:monopolist}).
	\item Perhaps surprisingly, we show that the tight revenue guarantees for revenue-ordered assortments against the TAOP obtained in \cite{berbeglia2020assortment} hold verbatim for the RAOP and even for the personalized RAOP where the firm can use a separate refine assortment for each consumer segment (Section~\ref{sec:monopolist}).
	\item We introduce three heuristics for the RAOP which are refinements of the traditional revenue-ordered assortments (Section~\ref{sec:heuristics}). We performed a series of computational experiments using the latent-class MNL (LC-MNL) instances available from the literature. Our numerical results show that these polynomial heuristic often yield higher expected revenues than the optimal solution for the TAOP (which is NP-hard) indicating that there are many situations where firms may benefit from a strategic use of the RAOP where the RAOP identifies the products that need to be refined, and once they are refined the problem can be solved as a TAOP until the prevailing conditions change  and new choice models are fitted at which point the RAOP may need to be called again.
	
\end{itemize}

Besides the main contributions stated above, we introduce the \emph{sequential assortment commitment problem} (SACP) in which a seller wishes to sell items to consumers over a finite time horizon by committing to an assortment schedule. We show that there is a strong connection between this problem and the RAOP. 

\section{Model}\label{sec:model}

A discrete choice model is a function that maps a mean utility vector $u = (u_0, u_1, \ldots,u_n)$ to non-negative numbers $(q_0(u), \ldots, q_n(u))$  such that $\sum_{i = 0}^n q_i(u) = 1$. We interpret $q_i(u)$ as the probability that the consumer selects product $i \in N: = \{1,\ldots,n\}$, and $q_0(u)$ as the probability that the consumer selects the outside alternative. We will assume without loss of generality that $u_0 = 0$, so from now on  when we refer to the vector $u$ we will drop the component corresponding to the outside alternative. A discrete choice model is said to be {\em regular} if $q_i(u)$ is decreasing in $u_k, k \neq i$.  The class of regular discrete choice models includes the class of random utility models (RUM). The representative agent model (RAM) assumes that $q(u)$ is as a solution to $S(u) := \max_{q \in \cQ} [u'q - C(q)]$, where $\cQ$ is the simplex, and $C$ is a function that penalizes the concentration of probabilities \citep{spence1976product,hofbauer2002global}. \cite{feng2017relation} shows that the RAM  contains the class of all RUMs and that $q(u)$ is regular when $S$ is sub-modular.



Let $\{0,1\} \subseteq \Theta_i \subseteq [0,1]$ for all $i \in N$, and let $\Theta : = \Theta_1 \times \ldots \times \Theta_n$. For fixed $u$  let $\tilde{u}_i : = u_i + \ln(x_i)$. For each $x \in \Theta$, let $p(x): = q(\tilde{u}) = q(u + \ln(x))$. Notice that if $q$ is regular then $p_i(x)$ is decreasing in $x_k, k \neq i$. Let $r$ be the unit profit contribution vector. We will refer to $r$ as the vector of revenues,  keeping in mind the broader interpretation as the profit contribution vector.  Let $R(x) := r'p(x)$. The refined assortment optimization problem (RAOP) is given by:
\begin{equation}
	\label{eq:raop}
	\cR(u|\Theta): = \max_{x \in \Theta} R(x).
\end{equation}
The case $\Theta_i = \{0,1\}$ for all $i \in N$ reduces to the traditional assortment optimization problem (TAOP), while $\Theta_i = [0,1]$ for all $i \in N$ is the fully flexible RAOP.  For economy of notation, we write $\cR^*(u)$ as the optimal expected revenue for TAOP and $\bar{\cR}(u)$ for the fully flexible RAOP.  We will write $\cR^*$ and $\bar{\cR}$ except when we need to make the dependence on $u$ explicit.

When there are multiple customer types, the firm's objective is to maximize $R(x) = \sum_{j \in M}\theta_j R_j(x)$ where  $M: = \{1,\ldots,m\}$ is the set of consumer types, and $\theta_j$ is the proportion of type $j$ customers. If the firm can engage in personalized RAOP (p-RAOP) the firm can earn $\bar{\cR}^p = \sum_{j \in M} \theta_j \bar{\cR}_j$ where $\bar{\cR}_j$ is the optimal expected revenue for type $j$ customers. Assume without loss of generality that the products are sorted in decreasing order of the $r_i$s, so $r_1 \geq r_2 \geq \cdots r_n \geq r_{n+1} := 0$.  Let $e_k$ be the $k$th unit vector and for each $i \in N$, define $e^i := \sum_{k \leq i}e_k$. Let $\cR^o := \max_{i \in N} R(e^i)$ be the maximum expected revenue among the class of revenue-ordered assortments.
Clearly
$$\cR^o \leq \cR^*  \leq \cR(u|\Theta) \leq \bar{\cR} \leq \bar{\cR}^p.$$

As we will see, RAOP can significantly improve revenues for the firm, and in some cases it can also increase consumer surplus.

\begin{example}
	Consider a firm that has two products to offer with profit contributions $r_1=3.5, r_2 = 1$. There are two customer types, each following a maximum utility model.  Type 1 customers have utility vector $u^1 = (1,5, 1.6)$, while Type  2 customers  have utility vector $u^2=(-1, 1.6)$. Let $\theta_1 = .3, \theta_2 = .7$ be the market shares of the two customer types. The optimal solution for the TAOP is $x = (1,0)$, resulting in $p(x) = (\theta_1, 0)$,  $\cR^* = 3.5\times 0.3 =1.05$, and expected consumer surplus $0.45$, with  type 2 customers left out. The RAOP with $\Theta_1 = \{0,1\}$ and $\Theta_2 = \{0, 0.8, 1\}$ results in $x = (1,0.8)$, $p(x) = (\theta_1, \theta_2)$, $\bar{\cR} = 1.75$, and expected consumer surplus $1.41$. This represents a 66.6\% improvement in profits for the firm and a 214.7\% increase in consumer surplus. All of the gains come from type 2 customers that buy product 2 at $x_2 = 0.8$, a level that does not cannibalize the demand for product 1 from type 1 customers.
\end{example}

\begin{example}
	Suppose a firm has three products to offer with profit contributions $r_1=100, r_2 = 65, r_3=58$. Consumers follow a latent class MNL model with 2 segments of equal weight (i.e. $\theta_1=\theta_2=0.5$). Let $u^1 = (0.01,100,0.1)$ and $u^2 = (100,1000,0.1)$ denote the mean utility vector for segments 1 and 2 respectively. The optimal solution for the TAOP is $x = (1,1,0)$,  $p(x) ~= (0.045, 0.95, 0)$,  and $\cR^* ~= 66.24$. The RAOP with $\Theta_1 = \Theta_3 = \{0,1\}$ and $\Theta_2 = [0,1]$. results in $x = (1,0.06,1)$, $p(x) ~= (0.311, 0.608, 0.007)$, and $\bar{\cR} ~= 71.06$. This represents an improvement of over 7\% in profits for the firm. Note that product 3 was added in full under the RAOP. The RAOP may refine products taken in full by the TAOP and fully add products rejected by the TAOP.
\end{example}


Notice that the flexibility afforded by the RAOP may have an impact in unit costs as a result in the change in utility. In many applications we expect the change in unit costs to be negligible. This is often true when modifying services, where imposing purchase or travel restrictions have zero or close to zero marginal costs. For applications where the change in unit costs are not negligible, we assume that any changes in costs are passed to the consumer so that the unit revenues (profit contributions) are unchanged.  More precisely, if the firm modifies the utility of product $i$ by $\delta u_i$  at a cost of $\delta c_i$, we assume that the firm adjusts prices by $\delta c_i$, resulting in $\delta r_i = 0$. Then $\tilde{u}_i = u_i + \delta u_i - \beta_i \delta c_i$, where $\beta_i$ is the price-sensitivity of product $i$. The price adjustment $\delta c_i$ may be suboptimal, but it is a reasonable approximation if the passthrough rate is close to one and/or $\delta c_i/c_i$ is small. As  TAOP, the RAOP is a lower bound on what the firm can make if it had more pricing freedom.

We end this section by introducing the \emph{Sequential Assortment Commitment Problem} (SACP) as a special case of the RAOP. In the SACP a firm commits to sequentially offer, possibly empty, assortments $A_s \subset N$, $s \in T: = \{1,\ldots,t\}$. The objective of the firm is to maximize expected profits from consumers who follow a discrete choice model and have mean valuations $u_{is}, i \in N, s \in T$ with $u_{is}$ decreasing in $s$. Unlike recent sequential assortment problems proposed where consumers are impatient and buy the first satisfying product they find  \citep{chen2019sequential,fata2019multi,flores2019assortment,liu2020assortment}, the SACP considers forward-looking consumers who optimally time their purchase to maximize their utility. The SACP is a special case of the RAOP as the firm can solve a RAOP with $\Theta_i = \{u_{is}, s \in T\}~~\forall ~~i \in N$ to obtain a solution for the SACP. 
\section{Tight bounds}\label{sec:monopolist}
The purpose of this section is two fold. First, we will show that the well-known revenue-ordered assortments heuristic possesses the same revenue guarantees for the RAOP as for the TAOP. Surprisingly, these revenue guarantees (which are already tight under the TAOP) also apply to the \emph{personalized} version of the RAOP (p-RAOP) in which the firm can customize the refine assortments to each of the consumer segments. Second, we present tight bounds on how much more revenue the firm can earn under the RAOP with respect to the TAOP for regular choice models, the MNL, the LC-MNL, and the Random Consideration Set model.

\subsection{Revenue Ordered Assortments for the RAOP: Performance Guarantees}
Before showing a tight bound on the performance guarantees of the revenue-ordered assortment heuristic, we need the following technical lemma.
\begin{lemma} \label{lem:reg_sales}
	Let $x \in [0,1]^n$ and $y \in \{0,1\}^n$. If $p$ is regular, then $y'p(x) \leq y'p(y)$.
\end{lemma}

\begin{proof}
Let $P(x) := \sum_{i \in N} p_i(x) = 1 - p_0(x)$. By regularity, $p_0(x)$ is decreasing in $x$, implying that $P(x)$ is increasing in $x$.  Then
$$y'p(\min(x,y)) = P(\min(x,y)) \leq P(y) = y'p(y),$$
so it is enough to show that $y'p(x) \leq y'p(\min(x,y))$ but this follows from regularity as $p_i(x) \leq p_i(\min(x,y))$ for all $i$ such that $y_i = 1$.
\end{proof}

\begin{theorem}
	\label{thm:upr}
	If $p^j$ is a regular discrete choice model for all $j \in M$, then  $\cR^* \leq \bar{\cR} \leq \bar{\cR}^p \leq \omega_n \cR^o \leq \omega_n \cR^*$,
	where $\omega_n :=n - (n-1)\alpha^{1/(n-1)}$ and $\alpha = r_n/r_1$.  Moreover,  $\omega_n$ is increasing in $n$, with $\omega_n \rightarrow 1 - \ln(\alpha)$ as $n \rightarrow \infty$.
\end{theorem}

\begin{proof}
The only non-trivial inequality is  $\bar{\cR}^p \leq \omega_n\cR^o$. Let $x^j\in [0,1]^n$ be an optimal solution of RAOP for segment $j$, and let $p^j_k(x)$ denote the probability that customer $j$ selects product $k$ under the refined assortment $x^j$.  Then

\bean
\bar{\cR}^p & = & \sum_{j \in M} \theta_j  \bar{\cR}_j \\
& = &  \sum_{i \in N} r_i \sum_{j \in M} \theta_j p^j_i(x^j) = \sum_{i \in N}(r_i - r_{i + 1})  \sum_{j \in M} \theta_j \sum_{k \leq i } p^j_k(x^j)\\
& \leq &  \sum_{i \in N}(r_i - r_{i + 1}) \sum_{j \in M} \theta_j \sum_{k \leq i } p^j_k(e^i)\\
& \leq &\sum_{i \in N}\frac{r_i - r_{i + 1}}{r_i} \sum_{j \in M} \theta_j  \sum_{k \leq i } r_kp^j_k(e^i)\\
& = &\sum_{i \in N}\frac{r_i - r_{i + 1}}{r_i} R(e^i) \leq  \sum_{i \in N}\frac{r_i - r_{i + 1}}{r_i} \cR^o\\
& \leq  &\omega_n \cR^o\\
\eean

%

The first inequality follows from Lemma~\ref{lem:reg_sales}. The second from the ordering of the $r_i$s and the third one from $R(e^i) \leq \cR^o$. The last inequality follows from maximizing $\sum_{i \in N}(r_i - r_{i+1})/r_i$ subject to $r_{n+1} = 0 < r_n \leq \ldots \leq r_1$ subject to $r_n/r_1 = \alpha$. This yields $r_k = r_1\alpha^{(k-1)/(n-1)}$ resulting in $\omega_n = \sum_{i \in N}(r_i - r_{i+1})/r_i = n - (n-1)\alpha^{1/(n-1)} = 1 + (n-1)[1- \alpha^{1/(n-1)}]$ as claimed. We leave it to the reader to verify that the term $\omega_n -1 = (n-1)[1 - \alpha^{1/(n-1)}]$ increases with $n$ and converges to $-\ln(\alpha)$.
\end{proof}

Theorem~\ref{thm:upr} extends  results by \cite{berbeglia2020assortment} from the TAOP to the p-RAOP. In the proof of Theorem~\ref{thm:upr} we allowed the $r_i$s to be all different. If the vector $r$ has only $k < n$ different values, the bound can be improved to $\omega_k = k - (k-1)\alpha^{1/(k-1)}$ with $\alpha$ unchanged  resulting in a tighter bound.  Similar arguments can be used to show that $ \bar{\cR}^p \leq \eta \cR^o$ where $\eta = (1 + \ln(Q_n/Q_1))$, $Q_1$ is the probability that product 1 sells under the optimal p-RAOP, and $Q_n$ is the probability that something sells under the same policy. All of these bounds were shown to be exactly tight relative to the TAOP so they are automatically tight for the  RAOP and for the p-RAOP.\\

\subsection{Bounds for Models that Satisfy the Monotone Utility Property}\label{sub:bounds_monotone}

We say that $R$  satisfies the {\em monotone-utility} property if $\cR^*(u) = \cR(u|\{0,1\}^n)$ is increasing in $u$.  This implies that $\bar{\cR} = \cR^*$, so refined assortment optimization is not useful when the  monotone-utility property holds.
\begin{lemmarep}
	\label{lem:mnlmup}The MNL satisfies the monotone-utility property.                                         
\end{lemmarep}

\begin{proof}
We need to show that if $x \in \{0,1\}^n$ is an optimal assortment, then $\cR^*(u) = \cR^*(u,x)$ is increasing in $u$. Set $v_i = \exp(u_i), i \in N$.  The choice probabilities can be written as $p_i(v,x) := v_ix_i/[v_0 + v'x]$
and the optimal expected revenue can be written as $\cR^*(v) := \sum_{i \in N} r_i p_i(v,x)$. Elementary calculus yields $\partial \cR^*(v)/\partial v_i = p_i(v,x)[r_i - \cR^*(v)]$ for all $i$ such that $x_i = 1$. Now  $r_i \geq \cR^*(v)$ for all $i$ such that $x_i = 1$ as otherwise it is optimal to eliminate $i$ from the assortment.  Thus $\cR^*(v)$ and therefore $\cR^*(u)$ is increasing in $u_i$ for all $i$ such that $x_i = 1$. If $x_i = 0$ then $r_i < \cR^*(v)$ and $\partial \cR^*(v)/\partial v_i = 0$ completing the proof.
\end{proof}


This implies that TAOP and RAOP will yield the same result for the MNL. Fortunately, $R (x)= \sum_{j \in M}\theta_j R_j(x)$ does not inherit the monotone-utility property from the $R_j, j \in M$. The next theorem gives an upper bound on $\bar{\cR}$, and the Proposition below shows that it is tight.

\begin{theorem}\label{th:mixed_m}
	If $R_j$ satisfies the monotone-utility property for all $j \in M$, then $\bar{\cR} \leq m\cR^*$.
\end{theorem}

\begin{proof}
Let $x^*$ be an optimal solution for the RAOP and let $\tilde{u}^j = u^j + \ln(x^*)$ for all $j \in M$. Then
\bean
\bar{\cR}(u) & = & R(x^*) \\
& = & \sum_{j \in M} \theta_j R_j(x^*) \leq  \sum_{j \in M} \theta_j \cR^*_j(\tilde{u}^j)\\
& \leq & \sum_{j \in M} \theta_j \cR^*_j(u^j) \leq  m \max_{j \in M} \theta_j \cR^*_j(u^j)\\
& \leq & m \cR^*.
\eean
The first inequality follows by solving a RAOP for each segment starting from $\tilde{u}^j = u_j + \ln(x)$. The second from the assumption that $\cR^*_j(\cdot)$ is an increasing function. The third is by picking the best segment and the fourth from $\max_{j \in M}\theta_j \cR^*_j(u) \leq \cR^*(u)$.
\end{proof}

By Lemma~\ref{lem:mnlmup}, Theorem~\ref{th:mixed_m} applies to the LC-MNL. The next result shows that the bound is arbitrarily tight for the LC-MNL. Since the LC-MNL is regular, we see that the bound from Theorem~1 is arbitrarily close by setting $m = n$.

\begin{propositionrep}\label{prop:tight_bound_LC-MNL}
	For every $n$ and every $m$, there exists a LC-MNL instance with $n$ products and $m$ segments such that $\bar{\cR}$ is arbitrarily close to $\min\{m,n\}\cR^*$.
\end{propositionrep}

\begin{proof}
We begin with the construction of a latent class MNL instance with $k\leq n$ consumer segments. Consider $\gamma<1$ and $\epsilon_1<1$ to be two small positive numbers. We set the nominal utility of product $i$ for consumer type $j$ to be $u^j_{i} = i(1+\epsilon_1)\ln(\gamma) + \beta_{ij} + \varepsilon_{ij}$ where $i(1+\epsilon_1)\ln(\gamma)$ is a segment independent utility component for product $i$, $\beta_{ij}$ is a segment dependent utility component; and $\varepsilon_{ij}$ are independent Gumbel random variables. For $j \geq i$ we set $\beta_{ij} = (j-i)\cdot \ln(\gamma)$. If $j<i$, $\beta_{ij}= M\ln(\gamma)$ where $M=\gamma^{-1}$ is a large positive number. The no-purchase alternative 0 has utility $u_0 = 3n\ln(\gamma)$.

Observe that when $\gamma$ is small enough, the product nominal utilities for consumer segment $j$ satisfy: $u^j_1 > u^j_2 > \hdots > u^j_j > u_0 > \Theta(M)\ln(\gamma)$ and $u^j_{h} = \Theta(M)\ln(\gamma)$ for all $h=j+1,\hdots,n$.

The revenue of product $i$ is set to $r_i = \frac{1}{\epsilon^{i-1}}$ for some $0<\epsilon<1$. Thus, the revenue increase as a function of the product index with $r_1=1$, $r_2=\epsilon^{-1}$, $r_3=\epsilon^{-2}$, etc.

Customer segment $j$ with $j<k$ has a probability mass of $\lambda_j = \epsilon^{j-1} - \epsilon^j$. Observe that $\sum_{j=1}^{k-1} \lambda_j = 1 - \epsilon + \epsilon - \epsilon^2 + \epsilon^2 - \epsilon^3 + \hdots + \epsilon^{k-2} - \epsilon^{k-1} = 1 - \epsilon^{k-1}$. Thus, $\epsilon^{k-1}$ is the probability mass of customer segment $k$.

Suppose we offer an assortment $S$ without refining the utility of the products. Let $\ell(S) = \min \{i: i \in S\}$ denote the product with smallest index in $S$. Consider now consumer segment $j$. When $\ell(S) \leq j$, we have that the probability that consumer type $j$ buys $\ell(S)$ when $\gamma$ tends to zero is:

$$\lim_{\gamma \to 0+}\mathcal{P}_j(\ell(S),S)\geq \lim_{\gamma \to 0+} \frac{\gamma^{j+\epsilon_1\ell(S)}}{\sum_{i=\ell(S)}^j\gamma^{j+\epsilon_1i} + \sum_{i=j+1}^n \gamma^{i(1+\epsilon_1)+\gamma^{-1}} + \gamma^{3n}}=1.$$

Above, we use the fact that $\mathcal{P}_j(i,S) \geq \mathcal{P}_j(i,\{1,\hdots,n\})$ and that to calculate the limit in the right we only need to determine the term that has the smallest exponent.
When $\ell(S) > j$, the probability that consumer type $j$ buys nothing when $\gamma$ and $\epsilon_1$ tends to zero is
$$\lim_{\gamma \to 0,\epsilon_1 \to 0+} \mathcal{P}_j(0,S) \geq \lim_{\gamma \to 0,\epsilon_1 \to 0+} \mathcal{P}_j(0,\{j+1,\hdots,n\}) = \lim_{\gamma \to 0} \frac{\gamma^{3n}}{\sum_{i=j+1}^n \gamma^{\gamma^{-1}} + \gamma^{3n}}=1.$$

Therefore in the limit, for any non-empty assortment $S \subseteq N$, $\ell(S)$ is the only product that has a non-zero probability of being purchased. Moreover, the probability of purchasing $\ell:=\ell(S)$ is $\sum_{j=\ell}^k \lambda_j = \epsilon^{\ell-1} - \epsilon^{\ell} + \epsilon^{\ell} - \epsilon^{\ell+1} + \hdots + \epsilon^{k-1}= \epsilon^{\ell-1}$. Thus, any assortment achieves the optimal revenue of $R^*= (\epsilon^{\ell(S)-1})\cdot \frac{1}{\epsilon^{\ell(S)-1}} = 1$.

Suppose now that the firm refines the segment-independent component of the product's $i$ utility from $i(1+\epsilon_1)\ln(\gamma)$ to $n\ln(\gamma)$ for every $i=1,\hdots,n$. Now the nominal utility for the customer segment $j$ becomes $\tilde{u}^j_i = (n+j-i)\ln(\gamma)$ if $j\geq i$. If $j<i$, the nominal utility for the customer segment $j$ becomes $\tilde{u}^j_i = (n+M)\ln(\gamma)$.

Thus, the nominal utilities for consumer segment $j$ now satisfy $\tilde{u}^j_j > \tilde{u}^j_{j-1} > \hdots > \tilde{u}^j_{1} > 0 > \tilde{u}^j_{j+1}= \hdots ~= \tilde{u}^j_{n} ~= \Theta(M)\ln(\gamma)$. In the limit when $\gamma \to 0$, by offering an assortment $S$ with $j \in S$ the probability that a consumer from segment $j$ buys product $j$ is:

$$\lim_{\gamma \to 0} \mathcal{P}_j(j,S) \geq \lim_{\gamma \to 0} \mathcal{P}_j(j,N) = \frac{\gamma^n}{\sum_{i=1}^j \gamma^{n+j-i} + \sum_{i=j+1}^n \gamma^{n+M} + \gamma^{3n}}=1.$$

As a result, if the firm offers the assortment $S=\{1,\hdots,k\}$ it can obtain
$$\sum_{j=1}^{k} \lambda_j \cdot r_j = \sum_{j=1}^{k-1} \frac{\epsilon^{j-1}-\epsilon^{j}}{\epsilon^{j-1}} + \frac{\epsilon^{k-1}}{\epsilon^{k-1}}=(k-1)(1- \epsilon) + 1 \geq k(1- \epsilon)R^*.$$

When $m\leq n$, setting $k=m$ shows that the bound is tight. If $m > n$, the bound is also tight since one can construct an LC-MNL instance with $m$ segments where the first $k=n$ segments are those in the above construction and the remaining $m-k$ segments are assigned a zero-weight.

\end{proof}
Since all RUM can be approximated by the LC-MNL our results imply that over that class of discrete choice models, the bound $\hat{\cR} \leq m \cR^*$ is arbitrarily close for all $m \leq n$. Since the MNL is a regular model, the bound for the previous section applies, so we conclude that for all RUMs, the bound $\hat{\cR} \leq \min(m,n) \cR^*$ is arbitrarily close. It is worth to mention that similarly to the TAOP, the RAOP under the LC-MNL is NP-hard \citep{desir2020} \footnote{In their study of the TAOP under the LC-MNL, \cite{desir2020} considered its continuous relaxation and showed that there is no efficient algorithm that can achieve an approximation factor guarantee of at least $O(\frac{1}{m^{1-\delta}})$ for any constant $\delta>0$ unless $NP \subset BPP$.}.

\subsection{Bounds for the Random Consideration Set Model}\label{sub:bounds_rcs}

We now provide tight bounds on the benefits of using RAOP with respect to TAOP when consumers follow the random consideration set (RCS) model introduced by \citet{manzini2014stochastic}. In this model, all consumers have the same preference ordering $1  \prec 2 \prec \ldots \prec n$, and product $i$ has attention probability $\lambda_i \in [0,1]$. If assortment $S$ is offered, then product $i$ is selected with probability
$$q_i(S) = \lambda_i \Pi_{j > i} (1-\lambda_j\delta(j \in S)),$$
where $\delta(j \in S) = 1$ if $j \in S$ and $0$ otherwise, and
products over empty sets are defined as 1.
\cite{gallego2017attention} showed that the RCS can be approximated exactly by a the Markov Chain model of \citet{blanchet2016markov} and developed an algorithm to find an optimal assortment that runs in $O(n)$ time. Here we consider the refined assortment optimization version of the RCS model. A reduction in the utility of product $i$ can result both in a reduction of the attention probability $\lambda_i$ and may bring product $i$ lower in the preference ordering. Those two issues need to be considered carefully in solving the RAOP for the RCS model. We show that the RAOP can at most double expected revenues relative to the TAOP under the mild assumption that attention probabilities are non-decreasing in the products utilities. This means that the RAOP under the RCS model has a 2-factor approximation algorithm.

\begin{theoremrep}
	\label{thm_RCS_factor_2}
	For every RCS, $\bar{\cR}\leq 2 \cR^*$.
\end{theoremrep}

\begin{proof}

We omit the proof that the best and the worse preference orderings are, respectively, in the order and in the reverse order of the revenues, and a lemma that shows that the revenue of the RCS is monotonic on the attention probabilities. The details are available from the authors upon request. We next proceed to the more difficult part of the result that bounds the expected revenue with the best order in terms of the expected revenue of the worst order.


Assume, that $r_i, i = 1,2, \ldots,n$ is a non-decreasing sequence of non-negative real numbers and let $\lambda_i, i = 1, 2, \ldots,n$ be an arbitrary sequence of attention probabilities in $(0,1)$. Suppose first, that the preference order is increasing in the index of the product: $1 \prec 2 \prec \hdots, \prec n$. Let $H_0 := 0$ and for $k \geq 1$,  let $H_k := H_{k-1} + \lambda_k(r_k -H_{k-1})$ for $k = 1, \ldots, n$.  Then $H_n$ is the optimal expected revenue under this preference ordering, and an upper bound on any other preference ordering.

We now compute the optimal expected revenue when the preference order goes in the opposite direction.   Set $G^n_0 = 0$ and for $k \in \{1,\ldots,n\}$ do  $G^n_k := G^n_{k-1} + \lambda_{n+1-k}(r_{n+1-k} - G^n_{k-1})^+$ and define $G_n: = G^n_n$. Then $G_n$ is the optimal expected revenue under this ordering, and a lower bound under any other preference ordering. Clearly $G_n \leq H_n$. We now show that
$$\frac{H_n}{G_n} \leq  2 ~~~\forall~~~n \geq 1.$$
Let $f(1) := 0$, $f(k) : = (1- \lambda_k)(\lambda_{k-1} + f(k-1))$ for $k > 1$. Define also the sequence $\hat{f}(1) = 0$, $\hat{f}(k) := (1- \lambda_{k+1})(\lambda_k+ \hat{f}(k-1))$ for $k > 1$. Notice that both $f$ and $\hat{f}$ are of the same form except for a shift in the index. We next show by induction that
$$f(k) =  \hat{f}(k-1) + \prod_{j=2}^k(1-\lambda_j) \lambda_1 ~~~~\forall~~~~k \geq 2.$$
Moreover, $f(k) \leq 1 - \lambda_k$, and $\hat{f}(k) \leq 1 - \lambda_{k+1}$ for all $k$.
For $k = 2$, the left hand side is $f(2) = (1-\lambda_2)\lambda_1$ while the right hand side is $\hat{f}(1) + (1-\lambda_2)\lambda_1 = (1- \lambda_2)\lambda_1$ so the result holds for $k = 2$. Suppose the result holds for $k$, so $f(k) =  \hat{f}(k-1) + \prod_{j=2}^k(1-\lambda_j) \lambda_1$.
Then
\bean
f(k+1)  & = &  (1-\lambda_{k+1}) (\lambda_k + f(k))\\
& = &  (1-\lambda_{k+1}) (\lambda_k + \hat{f}(k-1) + \prod_{j=2}^k(1-\lambda_j) \lambda_1)\\
& = &  (1-\lambda_{k+1}) (\lambda_k + \hat{f}(k-1)) + \prod_{j=2}^{k+1}(1-\lambda_j) \lambda_1) \\
& = & \hat{f}(k) +  \prod_{j=2}^{k+1}(1-\lambda_j)\lambda_1,
\eean
where the first equality follows from the definition of $f$, the second from the inductive hypothesis, the third by the distributive property, and the fourth from the definition of $\hat{f}(k) = (1-\lambda_{k+1}) (\lambda_k + \hat{f}(k-1))$. This completes the inductive step.

We next show by induction that $f(k) \leq  1 - \lambda_k$. This holds for $k = 1$ as $0 \leq 1 - \lambda_1$ follows from $\lambda_1 \leq 1$. Suppose it holds for $k-1$ so that $f(k-1) \leq 1 - \lambda_{k-1}$, then $f(k) = (1-\lambda_k)(\lambda_{k-1} + f(k-1)) \leq 1 - \lambda_k$ completing the proof. A similar argument applies for $\hat{f}$ and is omitted.

We will show that by induction that $H_n/G_n \leq 1 + f(n) \leq 2 - \lambda_n$. Assume the result holds for all instances of size $n-1$. Then the result holds for the instance of size $n-1$ that excludes product 1. Letting $\hat{H}_{n-1}$ and $\hat{G}_{n-1}$ be the optimal expected revenues corresponding to the best and the worst orderings of $\{2, \ldots,n\}$, we see that the inductive hypothesis correspinds to $\hat{H}_{n-1}/\hat{G}_{n-1} \leq 1 + \hat{f}(n-1)$. Consider now the instance of size $n$ that includes product 1 under the assumption that $r_1 \geq G_n$.  Then $r_1 \geq G_n \geq  G^n_{n-1}$ and $G_n = (1-\lambda_1)G^n_{n-1} + \lambda_1 r_1$.

Since $\hat{G}_{n-1} = G^n_{n-1}$ we see that
$$G_n=  (1- \lambda_1)\hat{G}_{n-1} + \lambda_1 r_1.$$
We can also write $H_n$ in terms of $\hat{H}_{n-1}$ resulting in
$$H_n=\hat{H}_{n-1} + \prod_{j=2}^n(1-\lambda_j) \lambda_1 r_1.$$
Since the ratio
$$\frac{H_n}{G_n} = \frac{\hat{H}_{n-1} + \prod_{j=2}^n(1-\lambda_j) \lambda_1 r_1}{(1- \lambda_1)\hat{G}_{n-1} + \lambda_1 r_1}$$
is decreasing in $r_1 \geq \hat{G}_{n-1}$, it follows that it is maximized by setting $r_1 = \hat{G}_{n-1}$. For this choice of $r_1$ we have

$$\frac{H_n}{G_n} = \frac{\hat{H}_{n-1} + \prod_{j=2}^n(1-\lambda_j) \lambda_1 \hat{G}_{n-1}}{\hat{G}_{n-1}}.$$

By the inductive hypothesis, $\hat{H}_{n-1} \leq (1+\hat{f}(n-1))\hat{G}_{n-1}$. Then,
\bean
H_n & \leq &  \left[1 + \hat{f}(n-1) + \prod_{j=2}^n(1-\lambda_j) \lambda_1\right]\hat{G}_{n-1}\\
& = & (1+f(n))G_n
\eean
where the equality follows the relationship between $\hat{f}$ and $f$, and $\hat{G}_{n-1} = G_n$. Dividing by $G_n$ we obtain that $H_n/G_n \leq 1 + f(n) \leq 2 - \lambda_n$. This completes the proof under the assumption that $r_1 \geq G_n$.

Suppose now that the assumption $r_1 \geq  G_n$ fails, and let $l = \arg\min \{i: r_i \geq G_n\}$, so $r_1 \leq \ldots \leq r_{l-1} < G_n \leq r_l$. We will now argue that the ratio $H_n/G_n$ can be upper bounded by an instance with $n$ products and $l = 1$.  Since the computation of $H_n$ involves the terms $r_1 \leq r_2 \leq \ldots \leq r_{l-1} < G_n$ and $H_n$ is increasing in all of the $r_i$s, it follows that $H_n$ is maximized, without changing $G_n$, by setting $r_1 = r_2 = r_{l-1} = G_n$ .  More formally, consider the modified instance where $\bar{\lambda}_i = \lambda_i, i \in \{1,\ldots,n\}$, $\bar{r}_i = r_i, i \geq l$ and $\bar{r}_i = G_n$ for all $i < l$. Clearly $H_n \leq \bar{H}_n$ while $G_n = \bar{G}_n$.
Therefore
$$\frac{H_n}{G_n} \leq \frac{\bar{H}_n}{\bar{G}_n} \leq 1 + f(n) \leq 2-\lambda_n$$
where the second inequality follows from the case $l = 1$. We remark that Wentao Lu\footnote{Personal Communication.} at HKUST independently obtained a different proof for the same bound the same week we did.

\end{proof}
Here we briefly summarize the key steps of the proof, and provide some managerial insights. First,  we show that if the preference orders are unchanged, then the optimal  expected revenue of the RCS model is increasing in the attention probabilities. Second, we show that if the products are  sorted in increasing order of  the $r_i$s, we can identify the best and the worst preference orderings in terms of the optimal expected revenue under traditional assortment optimization. The best preference order is $1 \prec 2 \prec \ldots \prec n$, and the worst is $n \prec n-1 \prec \ldots \prec 1$.  This result is intuitive as consumers have preferences for the products with the highest revenues in the former preference order and for the lowest in the latter. This suggest that refined assortment optimization is about reducing the attention probability of low revenue products to make them less attractive in the  preference ordering. We end this section by showing that the factor 2 bound is tight fort the RCS model.

\begin{proposition}
	\label{prop_RCS_tight_example}
	Theorem \ref{thm_RCS_factor_2} is tight.
\end{proposition}

\begin{proof}
Consider a RCS model with $n=2$, $\lambda_1=\epsilon$, $\lambda_2=1$, $r_1=\epsilon^{-1}$ and $r_2=1$ for some small $\epsilon>0$. Suppose that, if products are shown in full, consumers prefer product 2 over product 1: $1 \prec 2$. Thus, under the TAOP, the firm expected revenue is then $R^*= \max\{\lambda_2r_2 + (1-\lambda_2)\lambda_1 r_1, \lambda_1 r_1\} = \max\{1,1\}=1$.
Now suppose that if the firm slightly decreases the utility of product 1,  and as a consequence the consumer preference order is reversed to $2 \prec 1$ but the attention probabilities remain the same\footnote{Alternatively one can imagine that product 1 attention probability is reduced by some very small positive number and take the limit to zero.}. Then, we have that $\hat{R}^* = \lambda_1r_1 + (1-\lambda_1)\lambda_2r_2 = 1 + (1-\epsilon)1 = 2 - \epsilon$. The tightness follows by taking the limit $\epsilon \to 0$.
\end{proof}


\section{Refined Revenue-Ordered Heuristics}
\label{sec:heuristics}

In this section we propose three heuristics for the RAOP that are refined versions of the revenue-ordered heuristic. They all enjoy the guarantees from Theorem~\ref{thm:upr} as they are at least as good as the best revenue-ordered assortment. The first heuristic consists of finding the best revenue-ordered assortment in which the utility of the lowest revenue product offered is optimized. We called this, the \emph{Revenue Ordered Heuristic with one partial product}.

\begin{center} {\bf RO1: Revenue-Ordered Heuristic with one partial product} \end{center}
\begin{enumerate}
	\item For each $i \in N$ let $x^*_i: = \arg \max R(e^{i-1} + x_i e_i)$ and compute $R(e^{i-1} + x^*_ie_i)$.
	\item Let $k = \arg\max_{i \in N} R(e^{i-1} + x^*_i)$.
	\item Return $e^{k-1} + x^*_ke_k$ and $R(e^{k-1} + x^*_ke_k)$.
\end{enumerate}

This heuristic involves solving $n$ optimization problems over a single real variable over $[0,1]$. It is clear that $R(e^{k-1} + x^*_ke_k) \geq R(e^k)$ for all $k \in N$ so its performance is at least as good as the RO heuristic. The second heuristic builds upon the first by potentially adding more partial products.

\begin{center} {\bf RO2: Revenue Ordered Heuristic with several partial products} \end{center}
\begin{enumerate}
	\item For each $i \in N$, compute $x^*_{ii} = \arg \max R(e^{i-1} + x_ie_i)$
	\item For $k = i+1, \ldots, n$ compute $x^*_{ik} = \arg \max R(e^{i-1} + x^*_{ii}e_i + \ldots + x^*_{i,l-1}e_{l-1} + x_{ik} e_k).$
	\item Let $k = \arg \max_{i \in N} R(e^{i-1} + x^*_{ii}e_i + \ldots + x^*_{in} e_n)$
	\item Return $e^{k-1} + x^*_{kk}e_k + \ldots + x^*_{kn} e_n$ and  $R(e^{k-1} + x^*_{kk}e_k + \ldots + x^*_{kn} e_n)$.
\end{enumerate}

This heuristic requires to solve $O(n^2)$ optimization problems over a single $[0,1]$ variable. Clearly $R(e^{k-1} + x^*_{kk}e_k + \ldots + x^*_{kn} e_n) \geq R(e^{k-1} + x^*_k e_k)$ so this heuristic is the never worse than the RO1.

The third heuristic is similar to the previous one, but starting from a revenue ordered assortment, instead of evaluating sequentially in a revenue decreasing way among the candidates for being considered as partial, we greedily select the product that by partially adding it to the current solution, increases the revenue the most.  We call this extension the \emph{Revenue-Ordered Greedy-Heuristic with several partial products)}. This heuristic requires to solve $O(n^3)$ optimization problems over a single $[0,1]$ variable, and for that reason we refer to it as \textbf{RO3}. Observe that while is clear that this heuristic is better than RO1 it can be worst than RO2 since the greedy nature might make the algorithm to select a product in the sub-routine that in hindsight was not the best choice. The performance of these heuristics is reported in Section~\ref{sub:numerical_MMNL}.

\subsection{Upper Bounds for the LC-MNL model}\label{sub:upper_LC-MNL}

In this section we propose two easy to compute upper bounds for the LC-MNL. The LC-MNL is important because every random utility model can be approximated as accurately as desired by a LC-MNL. Since there is no polynomial algorithm to compute $\cR^*$ we present here two easy to compute upper-bounds that can help evaluate heuristics. The first one is the solution to the p-RAOP, $\bar{\cR} \leq \bar{\cR}^p = \sum_{j \in M} \theta_j \cR_j$ which is easy to compute since a RO assortment is optimal for each market segment. The second upper bound is based on writing the objective function of $\cR^*$ as a fractional program, which results in a bi-linear program that can be linearized to obtain the following linear programming upper bound:


$$\bar{\cR}^u := \max_{x,y,z} \sum_{j \in M} \theta_j \sum_{i \in N} r_i v_{ij} z_{ij}$$
subject to the constraints
\begin{eqnarray*}
	v_{0j}y_j + \sum_{i \in N}v_{ij}z_{ij} &=& 1~~~~\forall ~~j \in M \\
	z_{ij} &\geq& x_i~~~~\forall~~~i \in N,~~\forall~~j \in M \\
	z_{ij} & \leq & x_i/v_{0j}~~~\forall~~i \in N,~~\forall~~j \in M \\
	z_{ij} &\leq& x_i + y_j  -1 ~~~\forall~~i \in N,~~\forall~~j \in M\\
	z_{ij} &\geq& (x_i -1) /v_{0j} + y_j ~~~\forall~~i \in N,~~\forall~~j \in M \\
	0 &\leq& x_i \leq 1~~~\forall~~i \in N \\
	0& \leq& y_j~~~\forall ~~j \in M \\
	0 &\leq & z_{ij} ~~~\forall~~i \in N,~~\forall~~j \in M
\end{eqnarray*}

If the solution satisfies $z_{ij} = x_i y_j$ for all $i \in N, j \in M$ then the upper bound yields  an optimal solution. This is similar to the mixed-integer linear formulation studied in \citep{bront2009column,  mendez2014branch, feldman2015assortment}, but in our case $x_i$ is also allowed to take continuous values. \cite{csen2018conic} proposed a conic formulation  for the LC-MNL for the TAOP which relies, as we do above, on McCormick inequalities  \citep{mccormick1976computability}.

\newpage
\section{Numerical Results}\label{sub:numerical_MMNL}
In this section, we evaluate the performance of the heuristics proposed in Section~\ref{sec:heuristics}. 
Since any random utility model can be approximate  by a LC-MNL to any arbitrary precision \citep{mcfadden200mixednourl,chierichetti2018discrete} it is natural to study the performance of the heuristic on the LC-MNL.  We used synthetic instances of the LC-MNL for our experiments testing instances with $n \in \{5,10,15,50,100\}$ and $m \in \{2,5,10, 50,100\}$. To generate the product utilities, we applied a procedure used by \cite{berbeglia2018comparative} which depends on a parameter $\epsilon>0$ as follows:  for each customer class $j \in M$, we generate a random permutation $\sigma_j$ over the products including the outside option, that we associate to class $j$. The exponentiated utility of product $i \in N \cup\{0\}$ for a type $j$ customer is  modeled as:	$v_{ij} = \epsilon^{\sigma_j^{-1}(i)}$, 	where $\sigma_m^{-1}(i)$ denotes the position of product $i$ in the permutation $\sigma_j$.  The value of $\epsilon \in [0,1]$ measures how similar or different the utilities are for consumers of a specific class. A value of $\epsilon$ close to zero, implies different utilities are very different, whereas a value of $\epsilon$ close to $1$ makes them very similar.

Product prices are sampled from different distributions: uniform distribution $\mathcal{U}(1,10)$, Normal distribution $\mathcal{N}(50, 10)$, a multi-modal distribution, which is the result of adding two normally distributed variables with their means separated enough relative to their variance, exponential distribution $\text{Exp}(1)$ and skewed Normal $\text{Sk}(100)$.  Additionally, we consider the effect that a wider price spectrum could have on the performance of the heuristics. Since the bound obtained in Theorem~\ref{thm:upr}  is dependent on the ratio $\alpha = r_n/r_1$,  we scale the results obtained from the sampled prices, and set the minimum price to match $1$, and then the maximum price to be either $5$, $10$ or $100$ resulting in $\alpha \in \{.01, .1, .2\}$. 


For each instance we first assign prices at random to products, without considering their relation with product attractiveness across segments. We also consider instances where random prices are aligned with the exponentiated utilities and aligned in the opposite order, so products with higher prices have lower exponentiated utilities. 

For each  instance, we report the average (and the maximum) expected revenue for the three RAOP heuristics and the optimal revenue $\cR^*$ for the TAOP obtained through enumeration\footnote{We report this only for values of $n < 50$, as the solution by enumeration was too computationally expensive to compute.}. We present the performance of the heuristics relative to the best revenue-ordered assortment. Cells marked with an asterisk indicate that the heuristic outperformed the optimal solution for the TAOP, suggesting that polynomial heuristics for the RAOP can often improve over the TAOP without having to solve a difficult combinatorial problem.

\newcommand\fsize{.915}    

\begin{figure}[H]
	\centering
	\includegraphics[width=\fsize\textwidth]{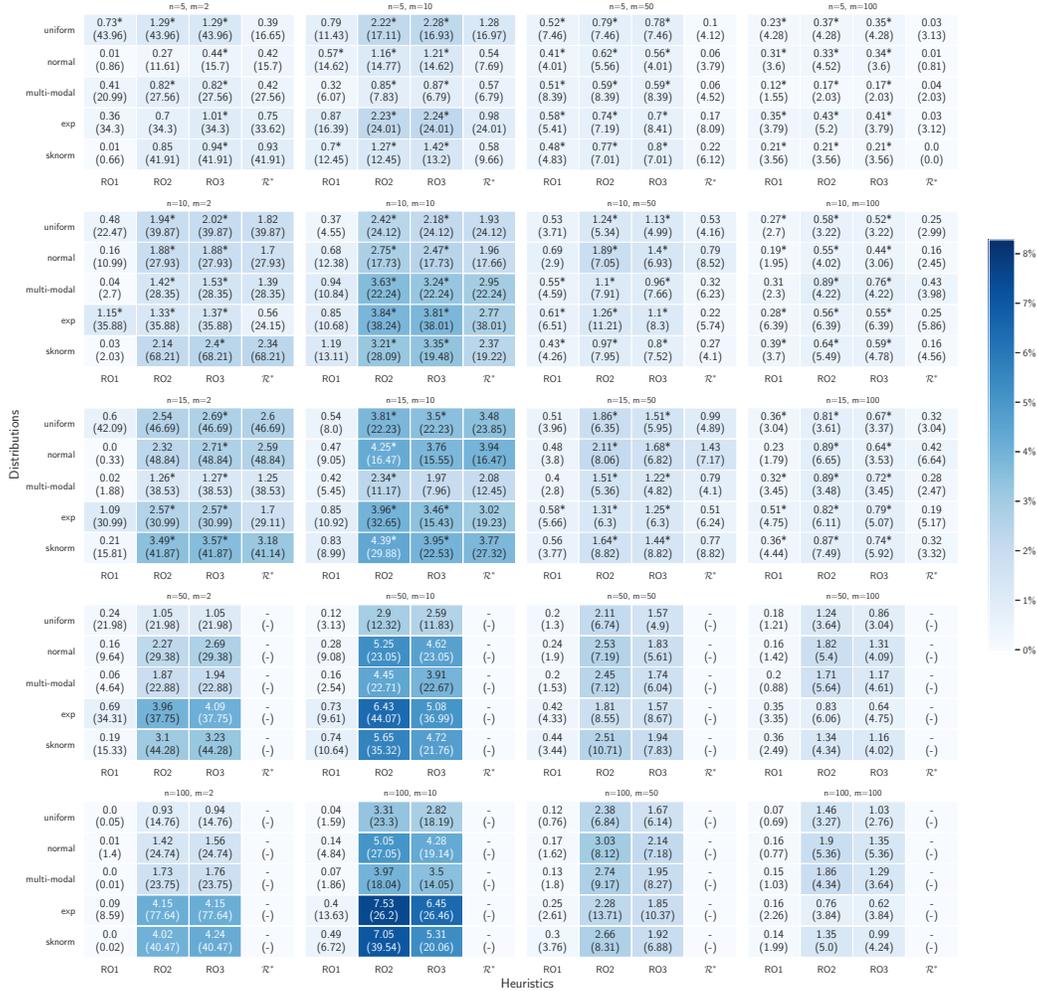}
	\caption{Mean (max) performance of heuristics and the optimal of TAOP relative to RO. For each heuristic, the value is highlighted with an asterisk if it outperformed the optimal of TAOP on average. }\label{fig1}
\end{figure}

Figure~\ref{fig1} shows the performance relative to the RO as a baseline for a set-up with high variation among products attractiveness ($\epsilon =0.01$) and $\alpha = 0.01$. This set up has tremendous variability in the utilities and a large ratio of the largest to the lowest value of $r$.  The heuristics consistently found refined assortments for the RAOP that outperform optimal assortments for the TAOP.  RO2 and RO3 present a growing gap with respect to RO as the number of products grows, while RO1 stays close to RO as it can only refine one product. When the number of customer types is large, the gap between RO and $\cR^*$ tightens, while the other heuristics significantly outperform $\cR^*$ on average, with the gap growing in $n$. The heuristics thrive in settings where there is higher concentration of products with high revenue since there are more opportunities of refining lower revenue products that cannibalize  products with higher revenues.  In addition to the observations concerning the average performance of the heuristic, we remark that the maximum over the instances can significantly outperform the RO heuristic and are often equal or better than the maximum for the TAOP.

\begin{figure}[H]
	\centering
	\includegraphics[width=\fsize\textwidth]{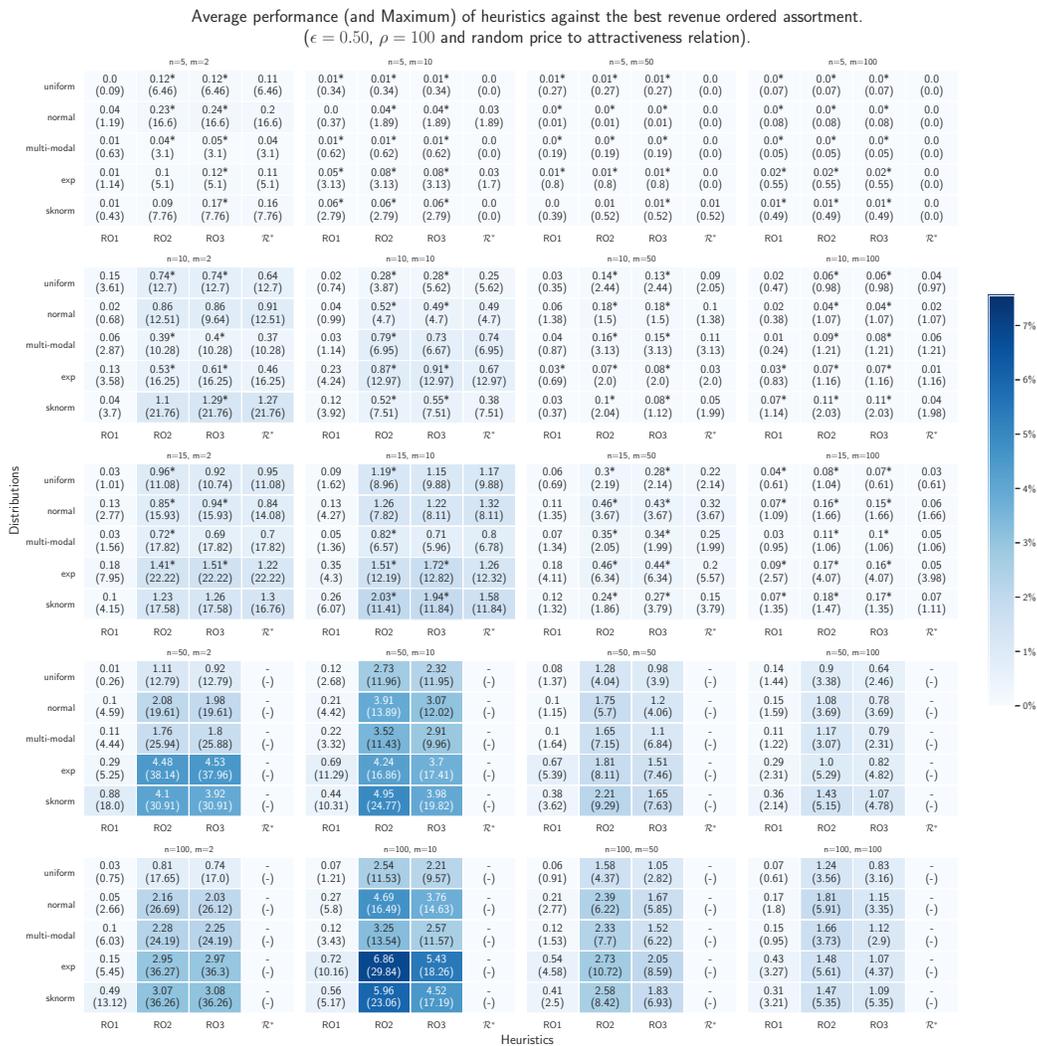}
	\caption{Mean (max) performance of heuristics and the optimal of TAOP relative to RO. For each heuristic, the value is highlighted with an asterisk if it outperformed the optimal of TAOP on average. }\label{fig2}
\end{figure}

Figure~\ref{fig2} shows results for cases with lower variation of the exponentiated utilities ($\epsilon = 0.5$). The heuristics  outperform the $\cR^*$  in most cases,  but the gap is tighter.  For higher values of $m$, the average and maximum uplift of the heuristics is slightly higher than in Figure~\ref{fig1}. For lower values of $m$, we see an increasing benefit of using the heuristics.


At a coarser level, Table~\ref{tb:avg_performance} shows that RO2 and RO3 perform on average better than $\cR^*$.  RO3 performs better with $5$ products, but for higher values of $n$ is on average dominated by RO2, which dominates $\cR^*$ across the board. We also evaluated the effect on market market share  and consumer surplus. RO2 had on average a mild increase of $0.38\%$ in market relative to $\cR^*$.  The difference in market share decreased with $n$ and grows with $m$. Note also that RO3, being a greedy heuristic, over-performed RO2 for $m=2$  across the board. This can be explained as it is less likely that the greedy heuristic goes wrong when $m$ is small. RO2 is more conservative doing modifications in strict order by revenue, which protects the current solution being assessed causing a better performance overall for $m>2$. RO2 is also computationally faster and recommended except when $m$ is small.

\begin{table}[h]
	\centering
	\label{tb:avg_performance}
	\caption{Average performance of heuristics and the optimal TAOP against revenue ordered assortment, across all instances, while varying $n$ and $m$. The maximum on each row is highlighted with an asterisk.}
	\begin{tabular}{|l|l|c|c|c|c|}
		\hline
		$n$ & $m$ & RO1    & RO2    & RO3    & $\cR^*$   \\ \hline
		5 & 2 & 0.12\% & 0.36\% & 0.39\%* & 0.27\% \\
		5 & 10 & 0.36\% & 0.84\% & 0.85\%* & 0.48\% \\
		5 & 50 & 0.2\% & 0.31\%* & 0.31\%* & 0.08\% \\
		5 & 100 & 0.11\% & 0.15\%* & 0.15\%* & 0.02\% \\
		10 & 2 & 0.15\% & 0.78\% & 0.83\%* & 0.7\% \\
		10 & 10 & 0.43\% & 1.78\%* & 1.62\% & 1.38\% \\
		10 & 50 & 0.28\% & 0.78\%* & 0.68\% & 0.36\% \\
		10 & 100 & 0.18\% & 0.44\%* & 0.39\% & 0.17\% \\
		15 & 2 & 0.14\% & 0.99\% & 1.03\%* & 0.91\% \\
		15 & 10 & 0.41\% & 2.32\%* & 2.07\% & 2.01\% \\
		15 & 50 & 0.3\% & 1.11\%* & 0.91\% & 0.66\% \\
		15 & 100 & 0.2\% & 0.61\%* & 0.5\% & 0.3\% \\
		50 & 2 & 0.1\% & 1.02\% & 1.04\%* & - \\
		50 & 10 & 0.28\% & 3.27\%* & 2.78\% & -\\
		50 & 50 & 0.22\% & 1.91\% *& 1.42\% & -\\
		50 & 100 & 0.17\% & 1.25\%* & 0.89\% & - \\
		100 & 2 & 0.06\% & 0.88\% & 0.91\%* & - \\
		100 & 10 & 0.21\% & 3.33\%* & 2.83\% & - \\
		100 & 50 & 0.16\% & 2.13\%* & 1.54\% & - \\
		100 & 100 & 0.13\% & 1.46\%* & 0.99\% & -\\
		\hline
	\end{tabular}	
\end{table}

In terms of consumer surplus it is clear that the heuristics will produce a higher consumer surplus than $\cR^*$ whenever they agree on the set of products that are offered in full. This is because consumers benefit from increased product availability under the RAOP. This agreement happens with relatively high frequency for the
RO2 and RO3  ($69.97\%$ and $68.83\%$ respectively for $n\leq 15$)  where their output offers the $\cR^*$ set in full. We also noticed that while this trend increases with the number of customer segments, it decreased with the number of products. If the firm is determined to use the RAOP to improve consumer surplus, it can first solve the $\cR^*$ heuristically or by enumeration and then determine which products left out can be brought in partially.

\section{Concluding remarks}

In this paper we proposed the refined assortment optimization problem and demonstrated that it can substantially improve revenues for the firm. Moreover, if the refinement is used only on products excluded by the traditional assortment optimization problem, then the expected consumer surplus goes up resulting in a win-win policy for the firm and its customers. We also developed refined revenue-ordered heuristics and showed that their worst case performance relative to the personalized refined assortment optimization has the same performance guarantees that were previously known relative to the traditional assortment optimization problem. For special demand classes, such as the multinomial logit, and the random consideration set model we showed that the benefits from personalized assortments relative to traditional assortment optimization have a factor of 1 and 2, respectively. Some interesting future research directions are: (1) quantifying the benefits of using RAOP with respect to TAOP in other important choice models such as the Exponomial and the Markov Chain model; (2) a detailed study of the RAOP in the case where $\Theta_i$ is a discrete set for each $i \in N$; (3) a best response analysis when a firm uses refined assortment optimization under competition; and (4) the study of the RAOP with cardinality constraints.

%

\newpage
\bibliographystyle{plainnat}
\bibliography{biblio}

\appendix
\renewcommand{\thesection}{\arabic{section}}

\end{document}